\documentclass[aps,pra,letterpaper,10pt,twocolumn,showpacs,superscriptaddress]{revtex4-1}
\usepackage{amsfonts,amssymb,amsmath,amsthm,graphicx}
\usepackage{url}
\usepackage{dsfont}
\usepackage{bbm}
\usepackage[usenames,dvipsnames]{xcolor}
 \usepackage{xcolor}
\usepackage{nicefrac}
\usepackage{natbib}
\usepackage{setspace}
\usepackage{subfigure}
\usepackage{tabularx,booktabs}
\usepackage{comment}
\usepackage[normalem]{ulem}
\usepackage{todonotes}
\usepackage{siunitx}
\usepackage{comment}
\usepackage[colorlinks=true, citecolor=blue,urlcolor=blue,linkcolor=blue,filecolor=blue]{hyperref}

\def\bra #1{\langle #1\vert}
\def\ket #1{\vert #1\rangle}

\def\virgolette #1{``#1"}

\newcommand{\beq}{\begin{equation}}
\newcommand{\eeq}{\end{equation}}

\newcommand{\ghz}{\si{\giga\hertz} }
\newcommand{\Pb}{p(\beta)}

\newtheorem{theorem}{Theorem}

\newtheorem{cor}[theorem]{Corollary}

\DeclareMathOperator{\Tr}{Tr}
\begin{document}

\title{Secure heterodyne-based quantum random number generator at 17 Gbps}

 \author{Marco Avesani}
 \affiliation{Dipartimento di Ingegneria dell'Informazione, Universit\`a degli Studi di Padova, Padova, Italia}

 \author{Davide G. Marangon}
 \affiliation{Dipartimento di Ingegneria dell'Informazione, Universit\`a degli Studi di Padova, Padova, Italia}
 
 \author{Giuseppe Vallone}
 \affiliation{Dipartimento di Ingegneria dell'Informazione, Universit\`a degli Studi di Padova, Padova, Italia}
 \affiliation{Istituto di Fotonica e Nanotecnologie, CNR, Padova, Italia}

 \author{Paolo Villoresi}
 \affiliation{Dipartimento di Ingegneria dell'Informazione, Universit\`a degli Studi di Padova, Padova, Italia}
 \affiliation{Istituto di Fotonica e Nanotecnologie, CNR, Padova, Italia}



 \begin{abstract}
Random numbers are commonly used in many different fields, ranging from simulations in fundamental science to security applications. In some critical cases, as Bell's tests and cryptography, the random numbers are required to be both secure (i.e. known only by the legitimate user) and to be provided at an ultra-fast rate (i.e. larger than Gbit/s).
However, practical generators are usually considered trusted, but their security can be compromised in case of imperfections or malicious external actions. 
In this work we introduce an efficient protocol which guarantees security and speed in the generation.
We propose a novel source-device-independent protocol based on  generic Positive Operator Valued Measurements and then we specialize the result to heterodyne measurements. The security of the generated numbers is proven
without any assumption on the source, which can be even fully controlled by an adversary. Furthermore, we experimentally implemented the protocol by exploiting 
heterodyne measurements, reaching an unprecedented secure generation rate of 17.42 Gbit/s, without the need to take into account finite-size effects.
Our device combines simplicity, ultrafast-rates and high security with low cost components, paving the way to new practical solutions for random number generation.
 \end{abstract}
\maketitle
\section{Introduction}
The possibility of generating random numbers by quantum processes is an invaluable resource in cryptography. Nowadays, common solutions based on Pseudo or classical Random Number Generators rely on deterministic processes, which are in principle predictable. On the contrary, Quantum Mechanics guarantees, from a theoretic point of view, that the outcome of the measurement is completely unpredictable. 
However, any imperfection in the physical 
realization of quantum random number generators (QRNG) can leak information correlated with the generated numbers, the so called \textsl{side information}.
Such classical or quantum correlations could be exploited by an eavesdropper to correctly guess the measurement outcomes.

The maximal amount of randomness that can be extracted in presence of such side information is given by the so called quantum conditional \textsl{min-entropy}~\cite{konig2009operational}.
Several approaches have been proposed to lower bound it, depending on the number of assumptions required on the devices used in the generator. 
For ``fully trusted'' QRNGs~\cite{Rarity1994,stefanov2000optical,jennewein2000fast},  the min-entropy can be evaluated because pure input states and well characterized measurement devices are assumed (see ~\cite{Vallone2014} for more details). In contrast, \textsl{device independent} (DI) protocols, by exploiting entanglement, don't need any assumption:
the violation of a Bell inequality directly bounds the min-entropy, without the need of trusting the generated state and the used measurement devices. Fully trusted QRNG, including all the commercial ones, are easy to realize, but they require strong assumptions for their use in cryptography.
On the contrary, DI protocols offer the highest level of security, but their realization is still too demanding for any practical use \cite{Pironio2010,Christensen2013,Bierhorst2017,Liu2018,Gomez2017}.

Semi-device-independent (Semi-DI) protocols \cite{Ma2016}, are a promising approach to enhance the security with respect to standard
``fully trusted'' QRNG, achieving fast generation rate, dramatically larger than DI-QRNG.
These require some weaker assumptions to bound the side information. Such assumptions can be related to the dimension of the underlying Hilbert space~\cite{Lunghi2014,Canas2014},   the measurement device \cite{Vallone2014,Marangon2017,Cao2016,Xu:16,Gomez2017} or the source \cite{Cao2015}, for example the mean photon number~\cite{Himbeeck2017semidevice} or the maximum overlap~\cite{Brask2017} of the emitted states.

In this work we introduce a QRNG belonging to the family of the Semi-DI generators.
In particular, we will describe a novel source-device independent (SDI) protocol 
by exploiting continuous variable (CV) observables of the electromagnetic (EM) field.
In previously realized CV-QRNGs \cite{Gabriel2010,Marangon2017},
random numbers were generated by using a homodyne detector that measures a quadrature of the EM field.
We propose and demonstrate a CV-QRNG based on heterodyne detection in the SDI framework: 
we will show how it is possible to obtain a lower bound on the eavesdropper quantum side information (i.e. the conditional min-entropy) and to achieve, 
to our knowledge, the fastest generation rate in the Semi-DI framework.

The advantages of heterodyne measurement over homodyne are multiple:
beside offering better tomography accuracy than homodyne
\cite{Rehacek2015, Muller2016}, heterodyne measurement offers an increased generation rate since it allows a 
``simultaneous measurement'' of both quadratures. 
In addition, the experimental setup is simplified with respect to the protocol based on homodyne introduced in \cite{Marangon2017}, since there is no need of an active switch to measure the two quadratures.
Finally, it is possible to derive a constant lower bound on the conditional quantum min-entropy, that doesn't change during the experiment.

Our SDI protocol assumes a trusted detector but it does not make any assumption on the source: an eavesdropper may fully control it, manipulating it in order to maximize her ability to predict the outcomes of the generator. Such approach is very effective in taking into account any imperfect state preparation.
Moreover, we will show the results of a practical realization of the protocol
with a compact fiber optical setup that employs only standard telecom components.
In this way, we are able to demonstrate a generation rate 
of secure random numbers 
greater than $17$ Gbit/s.

\section{A HETERODYNE QRNG}
In standard CV-QRNGs, random numbers are obtained
by measuring with an homodyne detector a quadrature observable of the EM fields, typically prepared in a vacuum state. 
CV-QRNGs are characterized by high generation rates due to
the use of fast photodiodes instead of (slow) single photon detectors: continuous spectrum of the observables typically assures more than one bit of entropy per measurement and the use of photodiodes with high bandwidth allow to sample the quadratures at GSample/s. 
In our QRNG, we implement a heterodyne detection scheme where two ``noisy quadrature observables'' are measured simultaneously \cite{Arthurs1965,Walker1987}. More precisely, an heterodyne measurement corresponds to the following Positive Operator Value Measurement (POVM) 
$\left\{\hat\Pi_\alpha\right\}_{\alpha\in \mathbb C}$
where
\beq
\hat\Pi_\alpha=\frac{1}{\pi}\ket{\alpha}\bra{\alpha} \,,
\eeq
 and $\ket{\alpha}$ is the coherent state with complex amplitude $\alpha$. If we define $\rho_A$ the density matrix of the EM field, the output of the 
heterodyne measurement is represented by the random variable $X$
\beq
X=\{\Re {\rm e} \left(\alpha\right),\Im {\rm m}\left(\alpha\right)  \} \,,
\eeq
distributed according to the following probability density function known as \textsl{Husimi function}:
 \beq
 Q_{\rho_A}(\alpha)=\text{Tr}\left[\hat\Pi_\alpha \rho_A \right]=\frac{1}{\pi}\bra\alpha\rho_A\ket\alpha\,.
 \eeq
In an ideal scenario where the QRNG user (Alice) can trust the source of random states: such scheme has the immediate advantage of doubling the generation rate with respect to an homodyne receiver.
Since the  \virgolette{raw} random numbers $\{\Re {\rm e}\alpha,\Im {\rm m}\alpha\}$ are typically not uniformly distributed, it is essential to  process them with a randomness extractor~\cite{Ma2013}. A randomness extractor compresses the input string of raw numbers, such that the shorter output string is composed by i.i.d. random bits.

In a fully-trusted QRNG, when the source is trusted and the input state is pure (such as for the vacuum) or the privacy of the generated numbers is not a concern, the number of random bits that can be extracted per sample is given by the so called classical min-entropy of $X$ 
\beq
\label{classHmin}
H_{\rm min}(X)=-\log_2[\max_\alpha Q(\alpha)] \,.
\eeq

However, ultrafast generation is worthless for cryptographic applications if the numbers are not secure and private. If security is important, quantum side information must be also taken into account and the conditional quantum min-entropy $H_{\rm min}(X|E)$ ~\cite{konig2009operational} must be evaluated. We recall that in the SDI framework an eavesdropper may have full control of the source and then may have some prior information on the generated numbers. 
We will show that with a heterodyne scheme it is possible to generate unpredictable and secure numbers also when the source of quantum states is controlled by the eavesdropper.

\section{A Secure heterodyne (or POVM) QRNG}
\subsection{Security of the CV protocol}
\label{sec:security}
In our SDI framework, Alice does not make any assumption on $\rho_A$, such as its dimension or purity: the source may be even controlled by a malicious QRNG manufacturer, Eve. 
This framework is well suited to deal with imperfect sources of quantum states~\cite{Vallone2014}.
On the contrary, Alice carefully characterizes her local measurement apparatus and trusts it.

In this scenario, Eve is assumed to prepare the state $\rho_A$ to be measured. 
In particular, Eve will prepare $\rho_A$ in order to maximize her guessing probability $P_\text{guess}$ of the outcomes of Alice heterodyne measurement. 
If the state $\rho_A$ is not pure, it can be prepared by Eve as a incoherent superposition of states $\hat\tau^A_\beta$ with probabilities $p(\beta)$, such as $\rho_A=\int \Pb\hat\tau^A_\beta d\beta$.
As shown below, for quantum state $\rho_A$ with positive Glauber-Sudarshan representation, Eve optimizes her strategy by using $\hat\tau^A_\beta$ that are coherent states.

When Eve generates the state $\hat\tau^A_\beta$, the best option for her is to bet on the heterodyne outcome with higher probability, namely  $\max_\alpha \Tr\left[\hat\Pi_\alpha \hat\tau^A_\beta \right]$. 
On average, Eve's probability of guessing correctly the output of the heterodyne measurement can be written as  
$P_{\rm guess}(X|\mathcal E)=\int  \Pb 
\max_\alpha 
\Tr \left[\hat\Pi_\alpha \hat\tau^A_\beta \right] d\beta$.
Having full control of the source, given the state $\rho_A$,  Eve chooses the  decomposition $\mathcal{E}=\{\Pb,\hat\tau^A_\beta\}$ 
that maximizes $P_{\rm guess}$.

\begin{figure}[ht]
\includegraphics[width=\linewidth]{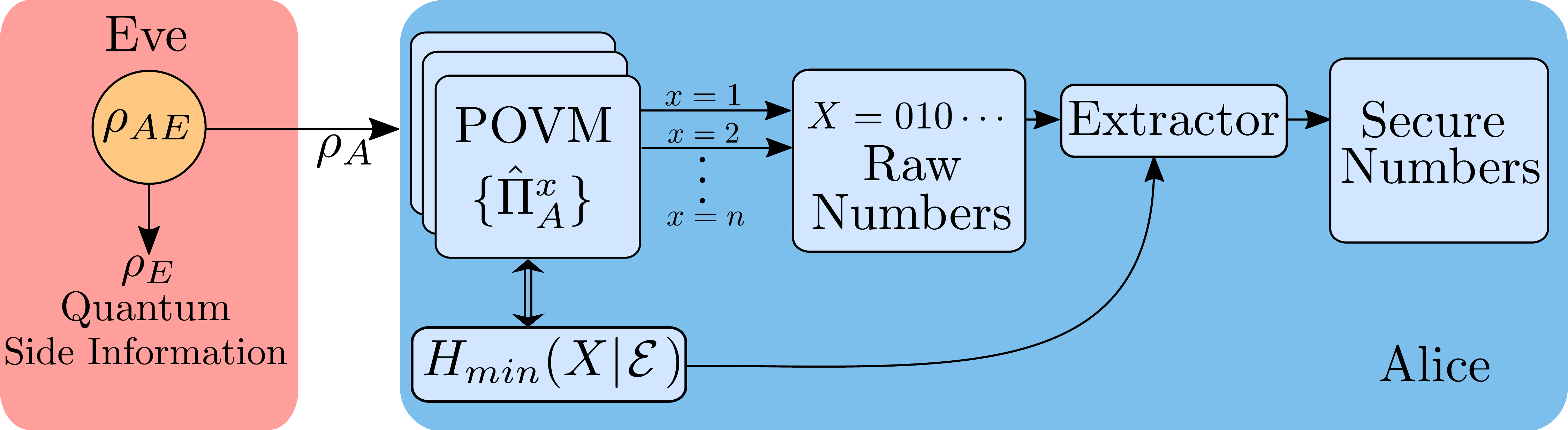}
\caption{In the general SDI scenario, Eve prepares the state $\rho_A$ that she sends to Alice such that her purification gives her the maximal guessing probability on Alice's outcome.  The structure of the POVM chosen by Alice to measure $\rho_A$ already impose a lower bound on $H_{min}(X|E)$, independently from the input state or the output of her measurement (see \href{prop1}{Proposition 1}). This bound is used to calibrate an extractor that returns secure random bits when applied to Alice's outcomes.}
\label{img:schemaPOVM}
\end{figure}

According to the Leftover Hash Lemma (LHL)~\cite{Tomamichel2011}, the extractable randomness in the presence of side information is quantified by the quantum conditional min-entropy
\beq
H_\text{min}(X|\mathcal{E}) = -\log_2 P_\text{guess}(X|\mathcal{E})\,,
\label{eq:cond_min_entropy}
\eeq
where $P_\text{guess}(X|\mathcal{E})$ is maximum probability of guessing $X$ conditioned on the quantum side information $\mathcal{E}$
\beq
P_\text{guess}(X|\mathcal{E})=\max_{\substack{\{\Pb,\hat\tau^A_\beta\}}} \int \Pb 
\max_\alpha
 \Tr\left[\hat\Pi_\alpha \hat\tau^A_\beta \right] d\beta \,.
\label{eq:gues_cond}
\eeq
The maximization in \eqref{eq:gues_cond} 
is performed over all
possible decomposition $\mathcal{E}=\{\Pb,\hat\tau^A_\beta\}$ that satisfy 
$\rho_A=\int \Pb \hat\tau^A_\beta d\beta$.
The above considerations are valid not only for the heterodyne measurement, but are correct for any POVM measurement (also with Hilbert spaces of finite dimensions).
Fig. \ref{img:schemaPOVM} represents a general protocol within this framework.
 It is worth noticing that 
$P_{\rm guess}$ is a true probability for finite dimension
Hilbert spaces, while it is a probability density for infinite dimension spaces (such as in the case of CV measurements).

By exploiting the properties of POVMs, we now derive a lower bound on 
$H_\text{min}(X|\mathcal{E})$ (and thus an upper bound on $P_\text{guess}(X|\mathcal{E})$).

\label{prop1}
{\bf Proposition 1.} {\it
For any POVM $\{\hat{ \Pi}_x\}$ the quantum conditional  min-entropy $H_{\rm min}(X|\mathcal{E})$ is lower-bounded by $-\max_{ \{ x,\hat\tau_A \in \mathcal{H}_A \} }\log_2(  \Tr[\hat{ \Pi}_x \hat\tau_A]  )$}.
\begin{proof}
\label{prop:POVM}
Given a set  of POVM $\{\hat\Pi_x\}$, the maximum over $x$ in 
\eqref{eq:gues_cond} is easy bounded by 
$\max_x
 \Tr\left[\hat\Pi_x \hat\tau^A_\beta \right]
 \leq\max_{x,\hat\tau_A}
 \Tr\left[\hat\Pi_x \hat\tau_A \right]
 $.
Then eq. \eqref{eq:gues_cond} is upper bounded by:
\begin{align}
\notag
P_{\rm guess}(X|\mathcal{E})_\text{min} & \leq  
\max_{x,\hat\tau_A}
 \Tr\left[\hat\Pi_x \hat\tau_A \right]
 \max_{\Pb,\tau_B} \int \Pb d\beta \\
& =\max_{ \{ x,\hat\tau_A \in \mathcal{H}_A \} } \Tr[\hat{ \Pi}_x \hat\tau_A]
\end{align}
from which the bound on the min-entropy easily follows by
using \eqref{eq:cond_min_entropy}.
\end{proof}
If the POVM reduce to projective measurements, the above bound is trivial, since it always possible to find a state $\hat\tau_A$ such that $\Tr[\hat{ \Pi}_x \hat\tau_A]=1$: in this case,
no randomness can be extracted.
However, for an overcomplete set of POVM we may have $\max_{ \{ x,\tau_A \} } \Tr[\hat{ \Pi}_{x} \hat\tau_A] < 1$ and therefore randomness can always be extracted.
We now exploit the above proposition for 
the specific case of heterodyne measurement.
\begin{cor}
\label{cor:het}
For the heterodyne measurement the quantum conditional min-entropy $H_{\rm min}(X|\mathcal{E})$ is lower-bounded by $\log_2\pi$. The bound is tight for quantum state with positive Glauber-Sudarshan $\mathcal P(\alpha)$ representation.
\end{cor}
\begin{proof}
It is well known that the Husimi function is upper bounded by $\frac1\pi$. 
Then $\Tr[\hat{ \Pi}_\alpha \hat\tau_A]=\frac{1}{\pi}
\bra{\alpha}\tau_A\ket{\alpha}=Q_{\tau_A}(\alpha)\leq\frac1\pi$, $\forall \tau_A$. 
By proposition 1, it follows that 
$H_{\rm min}(X|\mathcal{E})_\text{min}
\geq \log_2\pi$. To
 show the tightness, we note that any matrix $\rho_A$ can  be written as $\rho_A= \int \mathcal P(\alpha) \ket{\alpha}\bra{\alpha} d^2\alpha$ where $\mathcal P(\alpha)$ is the Glauber-Sudarshan P-function. 
If $\mathcal P(\alpha)$ is positive it can be interpreted as a probability density and the state $\rho_A$
can be seen as an incoherent superposition of coherent states. 
Since coherent states maximize the output probability for the heterodyne POVM, then the optimal decomposition in \eqref{eq:gues_cond} is precisely $\mathcal E=\{\mathcal P(\alpha),\ket{\alpha}\bra{\alpha}\}$ and the quantum conditional min-entropy is exactly $H_{\rm min}(X|\mathcal E)=\log_2\pi$.
\end{proof}

By using an heterodyne measurement scheme, a quantum tomography of the input state is also obtained \cite{leonhardt1997measuring}:
while Alice generates the raw random numbers, she also reconstructs the state $\rho_A$. Then it is possible to evaluate numerically the quantum conditional min-entropy 
by using \eqref{eq:cond_min_entropy} and \eqref{eq:gues_cond}.
Although for a qubit system, this problem was elegantly addressed by \cite{Fiorentino2007a}, it is not of easy solution in the CV case. 
On the other hand, Corollary 1 gives an easy lower bound on $H_{\rm min}(X|\mathcal E)$. 
Alice knows that even if Eve forges a state with an optimal $\mathcal{E}$, such side information will not let Eve guess the heterodyne outcome with a probability (density) larger than $\frac{1}{\pi}$.
In the presence on an imperfect source of quantum states, this is the most conservative strategy to adopt, but ensures the generation of completely secure random numbers while avoiding a complex numerical maximization.

It is worth noticing that
in many cases such lower bound is {\it tight}: indeed,
coherent and thermal states have positive 
Glauber-Sudarshan $\mathcal P(\alpha)$ function
and for those states the bound $\log_2\pi$ is tight.
Moreover, in contrast to other Semi-DI QRNG where the min-entropy needs to be estimated in real time to provide security \cite{Lunghi2014,Brask2017,Marangon2017}, in our protocol it depends on the structure of the heterodyne POVM and it is always constant.
Hence, Alice can apply on $X$ a randomness extractor calibrated on $\log_2\pi$ and erase any guessing advantage of Eve.
In the following we adapt the bound of Corollary 1 to realistic POVMs with finite resolution.
\subsection{Practical Bound}
\begin{figure*}[!t]
\includegraphics[width=\textwidth]{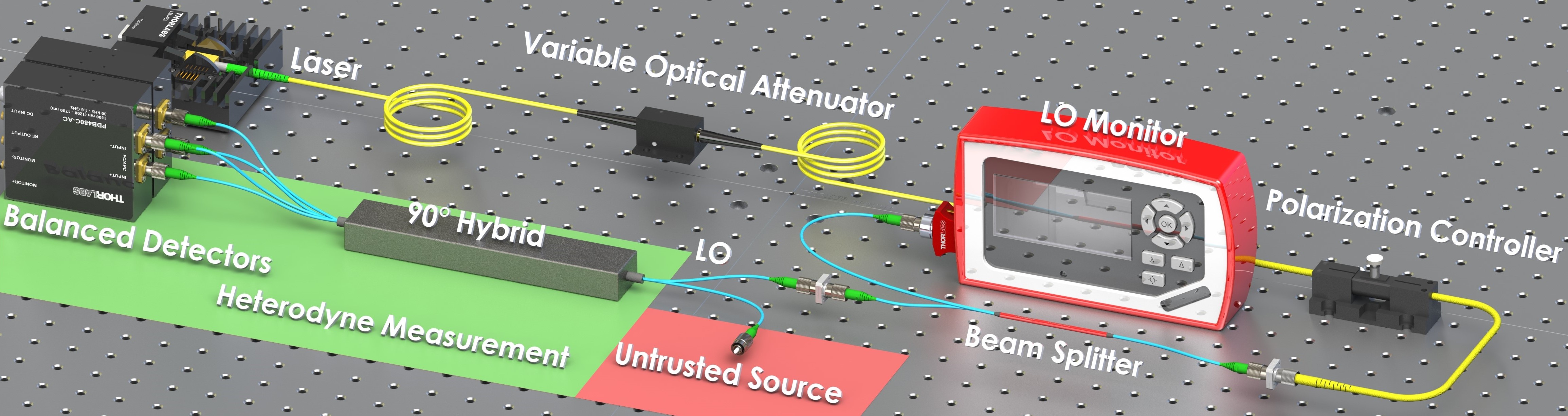}
\caption{Schematic representation of the experimental setup. Only commercial off-the-shelf devices were used.}
\label{img:schema3D}
\end{figure*}
In a real implementation, any heterodyne measurement is discretized. This means that the possible outcomes of the measure are discrete with a resolution
given by $\delta_q$ and $\delta_p$ for the two ``quadratures''. 
The discretized version of the POVM $\hat\Pi_\alpha$ is then given by
$
\hat \Pi^{\delta}_{m,n}=
\int_{m\delta_q}^{(m+1)\delta_q}\!\!d\Re{\rm e}\alpha
\int_{n\delta_p}^{(n+1)\delta_p}\!\!d\Im \text{m}\alpha
\,\,\hat\Pi_\alpha
$

and the average probability of guessing their outputs is given by
$P_\text{guess}(X|\mathcal{E})=
\max_{\{p_\beta ,\hat\tau^A_\beta \}} 
\int \Pb 
\max_{m,n}
\Tr\left[\hat\Pi^\delta_{m,n} \hat\tau^A_\beta \right]
$.

The term $\text{Tr}\left[\hat\Pi^\delta_{m,n} \hat\tau^A_\beta \right]$ is upper bounded by $\frac{\delta_q\delta_p}{\pi}$. 

Then the probability $P_\text{guess}$ is upper bounded by $\frac{\delta_p\delta_q}{\pi}$ and the quantum min-entropy is lower bounded by
\beq
\label{eq:bound}
H_\text{min}(X_\delta|\mathcal{E}) \geq \log_2\frac{\pi}{\delta_q\delta_p}\,.
\eeq
Hence, in the real-life implementation, the min-entropy of the random numbers is bounded by a function that depends on the measurement resolution only.
The measurement, in this scenario, is under control of the user: Alice can readily obtain the min-entropy (\ref{eq:bound}) by measuring $\delta_p$ and $\delta_q$ of her well characterized apparatus. 
The min-entropy is constant and Alice does not need to worry updating its value, as long as she trusts the apparatus.

\section{ Experimental results}
The proposed new protocol has been implemented with an all-fiber setup at telecom wavelength with the scheme in Fig. \ref{img:schema3D}; in this way is possible to exploit the availability of fast off-the-shelf components for classical telecommunication while keeping the setup compact.
The heart of the experiment lies in the heterodyne detection of the vacuum state, that  samples the $Q$ function with the help of a coherent field $\ket{\alpha}$ of a Local Oscillator (LO). We employed a narrow linewidth ECL laser at $1550 \si{\nano\meter}$ (Thorlabs SFL1550) followed by and electronically-controlled Variable Optical Attenuator (VOA) and a in-line Polarization Controller (PC). In this way we were able to finely control the intensity and the polarization of our LO, besides making the calibration procedure automatized.

Before entering the heterodyne measurement, $10\%$ of the LO is sent to a photodetector, for a continuous monitor of its intensity. By doing that, any anomaly to the normal functioning of the LO can be noticed in real-time, and  deviations can be compensated during the post-processing.

The optical heterodyne was realized with a commercial fiber integrated \virgolette{90 degree hybrid}: one port is coupled to the LO while from the other is entering the vacuum state. However, we work in the SDI scenario and, from the point of view of security, this port can be fully controlled by Eve, since we don't assume anything about the source.
The 90 degree hybrid mixes the signal with the LO and returns two pairs of outputs, featuring a $\nicefrac{\pi}{2}$ phase shift. These optical signals, detected by a couple of high-bandwidth balanced detectors ($1.6\, \ghz$ Thorlabs-PDB480C), are proportional to the quadratures of the signal, Re$\left[ \alpha \right]$, Im$\left[ \alpha \right]$.

We sampled both signals coming from the detectors using a fast oscilloscope with a sampling rate of $10$ GSamples/s at 10 bits of resolution (Lecroy HDO 9404). Each samples contains 20 bits of raw data, 10-bit for Re$\left[\alpha\right]$ and 10-bit for Im$\left[\alpha\right]$.
The raw signals of the ADC are proportional to the quadratures and directly sample the Q-function in the phase space, as shown in Fig. \ref{img:fig_exp}. The resolution of the ADC can be directly converted to the equivalent resolution in the phase space; in our case we got 
$\delta\text{Re}\left[\alpha\right]=(14.05 \pm 0.02)\cdot10^{-3} $ and $\delta\text{Im}\left[\alpha\right]=(14.14 \pm 0.02)\cdot10^{-3}$, respectively.

\begin{figure}[!ht]
	\begin{center}
		\includegraphics[width=\linewidth]{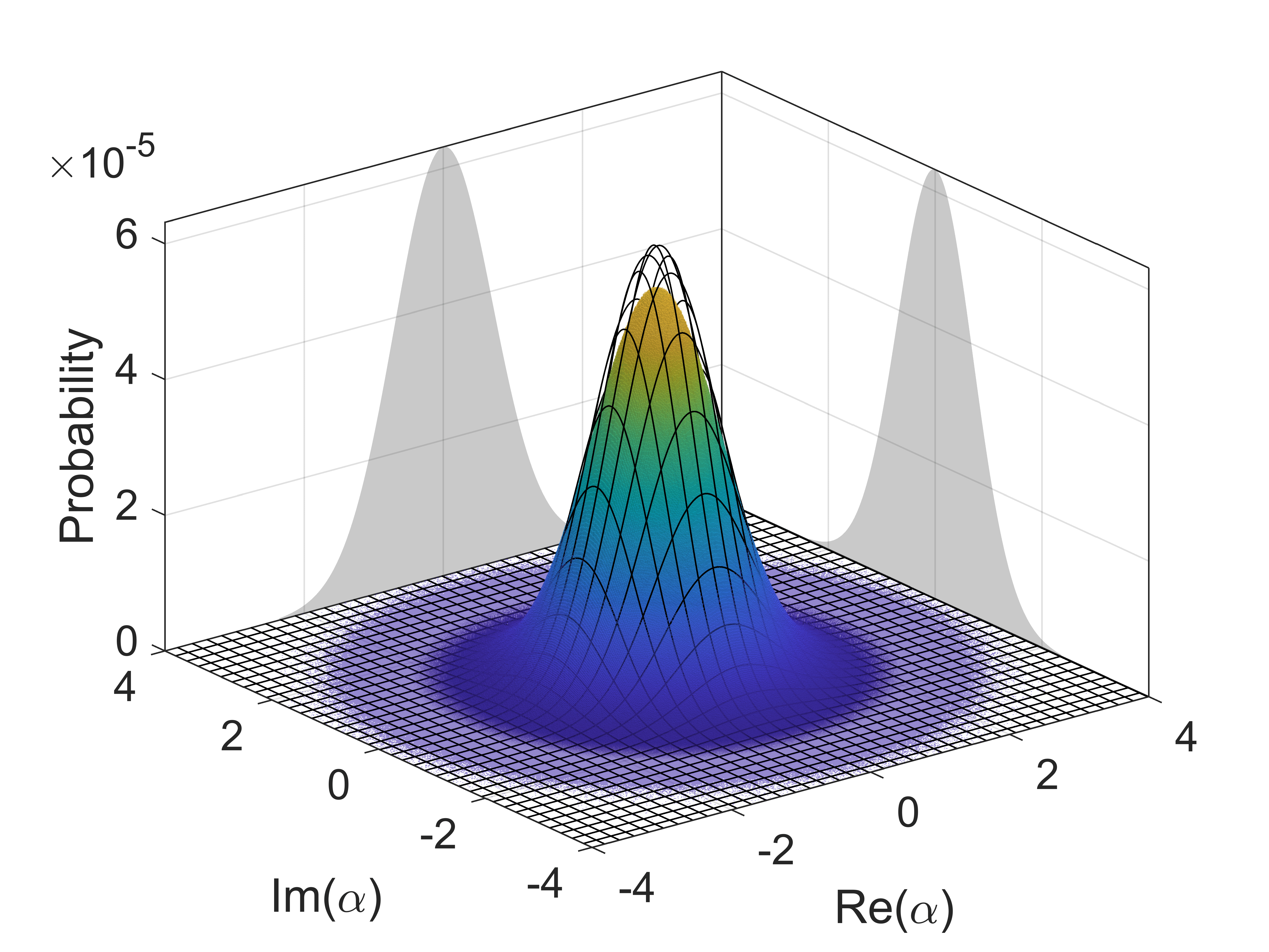}
	\end{center}
	\caption{The plot shows the Husimi function for the vacuum (meshed curve) and the measured state (colored histogram). The measured variance is slightly larger that the one expected for the vacuum due to the electronic noise that widens the distribution. }
	\label{img:fig_exp}
\end{figure}

The raw data are then digitally filtered, taking only a $1.25\, \ghz$ window in the central part of the spectrum obtained by the detectors. This is necessary in order to remove classical noise that is coupled with the detector. Finally, the data are downsampled at $1.25$ GSample/s, matching the bandwidth of the signal and removing any correlation introduced by the oversampling.

We acquired $6\cdot10^{10}$ measurements obtaining $\sigma^2_{\text{Re}\left[\alpha\right]}=0.55135 \pm 0.00001$ and $\sigma^2_{\text{Im}\left[\alpha\right]}=0.56732 \pm 0.00001$.
As it can be seen from Fig. \ref{img:fig_exp},  the measured Q-function is slightly larger than the one expected for a pure vacuum state, where both variances are expected to be equal to 1/2. 
The increase of the variances is due to classical noise of the detectors: in our approach, such noise is regarded as a ``spreading'' of the Q-function and thus is already included in our analysis for the quantum min-entropy.
The classical min-entropy $H_{\rm min}(X)$ corresponds to the larger probability of output and it is given by
\beq
\label{eq:HminClassical}
H_{\rm min}(X)=14.100\,.
\eeq
However, the quantum min-entropy can be lower bounded by eq. \eqref{eq:bound}. With the quadrature resolutions used for the experiment, we obtain
\beq
H_{\rm min}(X|\mathcal{E}) \geq 13.949\,,
\label{eq:hmin}
\eeq
for an equivalent secure generation rate of 17.42 Gbit/s. 
It is worth noticing that the high gain in security guaranteed by the conditional quantum min-entropy of eq. \eqref{eq:hmin} with respect to the classical min-entropy eq. \eqref{eq:HminClassical}
implies a very small reduction of the generation rate (from 14.10 to  13.949 bits per sample).

In addition, these rates are not calculated in the asymptotic regime, i.e. in the limit of infinite repetitions of the protocol, but are valid for single shot measurements.
In fact, the conditional min-entropy $H_{\text{min}}(X|\mathcal{E})$ is not estimated from the data, but it's bounded considering the structure of the POVM and the optimal strategy for the attacker, making it independent from the number of rounds of the protocol.
Finally, a Toeplitz randomness extractor \cite{Frauchiger2013} is calibrated using $H_{\rm min}(X|\mathcal{E})$, and extracts the certified numbers from the raw data. As a final check, we applied a series of statistical tests from the DieHarder and NIST suite: all of them are successfully passed, as shown in Appendix \ref{sec:statistica_tests}.

\section{Conclusions}
In this work we demonstrated the versatility of heterodyne detection scheme for the generation of secure random numbers in a CV-SDI framework, where no assumption on the source of quantum state is required.
In fact, exploiting the properties of the POVM implemented by the heterodyne measurement,
in Corollary 1 we obtained a direct lower bound to the conditional min-entropy, and hence on its security.
This bound enables the user to erase all the side information related with an imperfect or malicious source of quantum states.
Compared to previous SDI-QRNGs \cite{Vallone2014,Ma2016,Marangon2017} this security is obtained 
without affecting the generation rate: in the previous protocols, 
part of the generated numbers were consumed to estimate and update the bound to the conditional min-entropy.
In the protocol introduced here, the bound is constant, since it is determined by the resolution of the trusted measurement apparatus only.
Hence, \textsl{all} the secure numbers are available to the user.
Such simplification has many advantages for any practical implementation of the protocol.
Our approach allows indeed to merge the speed of heterodyne measurements and the security of semi-device-independent protocols.
Indeed, we realized the protocol with off-the-shelf components achieving 17.42 Gbit/s rates,
which is to our knowledge the fastest random generation rate for a semi-DI QRNG.

\begin{acknowledgments}
The authors thank R. Filip for fruitful discussions.
\end{acknowledgments}
\appendix 
\section{Calibration}
In the SDI framework we assume a trusted and characterized measurement device. In order to enforce that, before every run of the experiment we perform a calibration of our detection stage. This procedure is necessary for the evaluation of security, because it links the voltage output of the detectors to the relative quantities in the phase space, enabling us to calculate $\delta_q,\delta_p$.
\begin{figure}[!ht]
	\begin{center}
		\includegraphics[width=\linewidth]{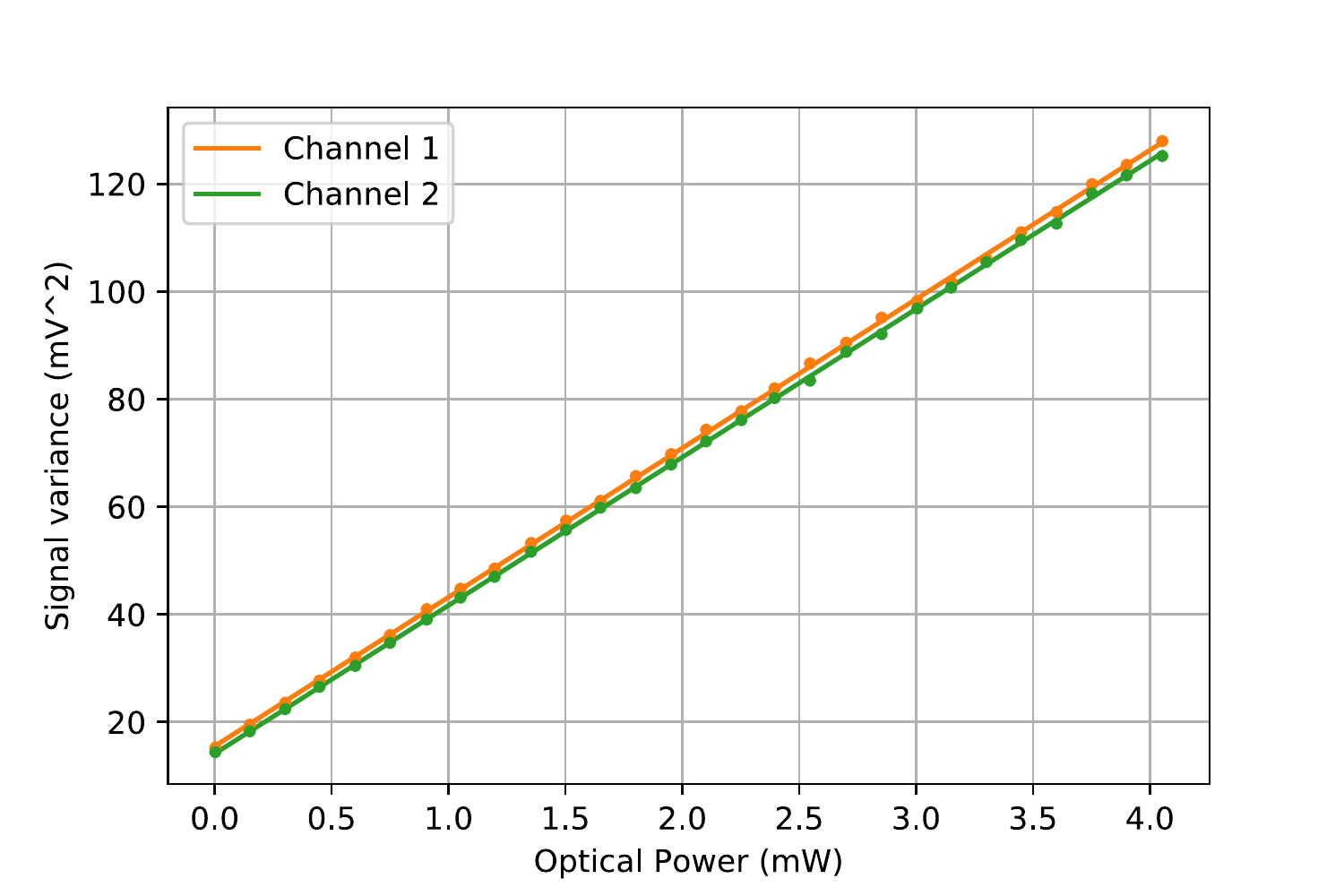}
	\end{center}
	\caption{The graph shows the linear dependence of the signal quadrature $\sigma^2_V$ as a function of the LO power.}
	\label{img:calibrazione}
\end{figure}
\begin{figure*}[!ht]
\includegraphics[width=0.98\textwidth]{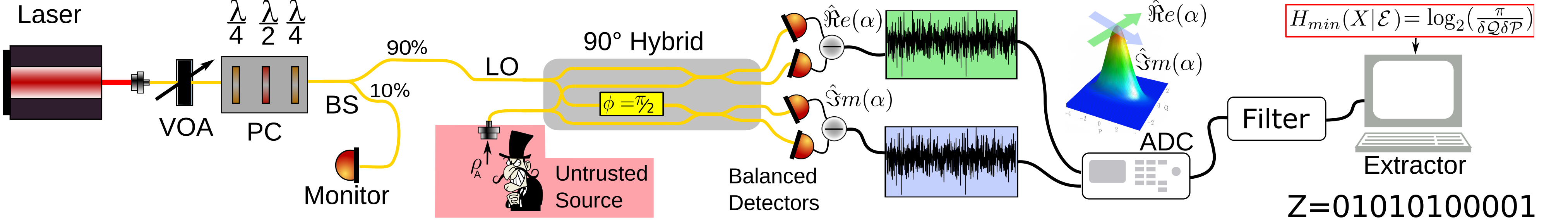}
\caption{Schematic representation of the experimental setup. The elements present are: Laser source used for the Local Oscillator (LO), Variable Optical Attenuator (VOA), Fiber Polarization Controller (PC), Fiber Beam Splitter (BS), a 90$\deg$ Optical Hybrid, a couple of High-Speed balanced photodetectors, a fast oscilloscope used as an Analog-to-Digital Converter (ADC), PC for the digital filtering and extraction }
\label{img:schema2D}
\end{figure*}

The calibration is performed automatically by the software that controls the QRNG: by varying the Variable Optical Attenuator (VOA), the power of the LO is changed from  $\num{0.01}\si{\milli\watt}$ to $4.05 \si{\milli \watt}$, when measured with the monitor photodiode.  
For each power, the signal of the balanced detector is recorded and the variance  $\sigma^2_V$ is estimated. As we can see in Fig. \ref{img:calibrazione} the relation is linear for all the tested powers (i.e. we never reached the saturation of the detector's amplifiers). From the fit, $m_1=(2.783 \pm 0.005)\cdot 10^{-2} \si{\volt^2\per \watt}$ and $q_1=(1.526 \pm 0.005)\cdot10^{-5} \si{\volt^2}$ for the slope and intercept of the first detector and $m_2=(2.748 \pm 0.004)\cdot 10^{-2} \si{\volt^2\per \watt}$ and $q_2=(1.419 \pm 0.004)\cdot 10^{-5} \si{\volt^2}$ for the second one. The slopes were used to convert
the measured voltages into phase-space quantities.The non-null intercept in both cases is caused by the electronic excess noise from the detectors and, since does not originate from the quantum measurement, is regarded a side-information available to Eve.

\section{Filtering, noise and autocorrelation}
\vspace*{-5mm}
\begin{figure}[!htbp]
	\begin{center}
		\includegraphics[width=\linewidth]{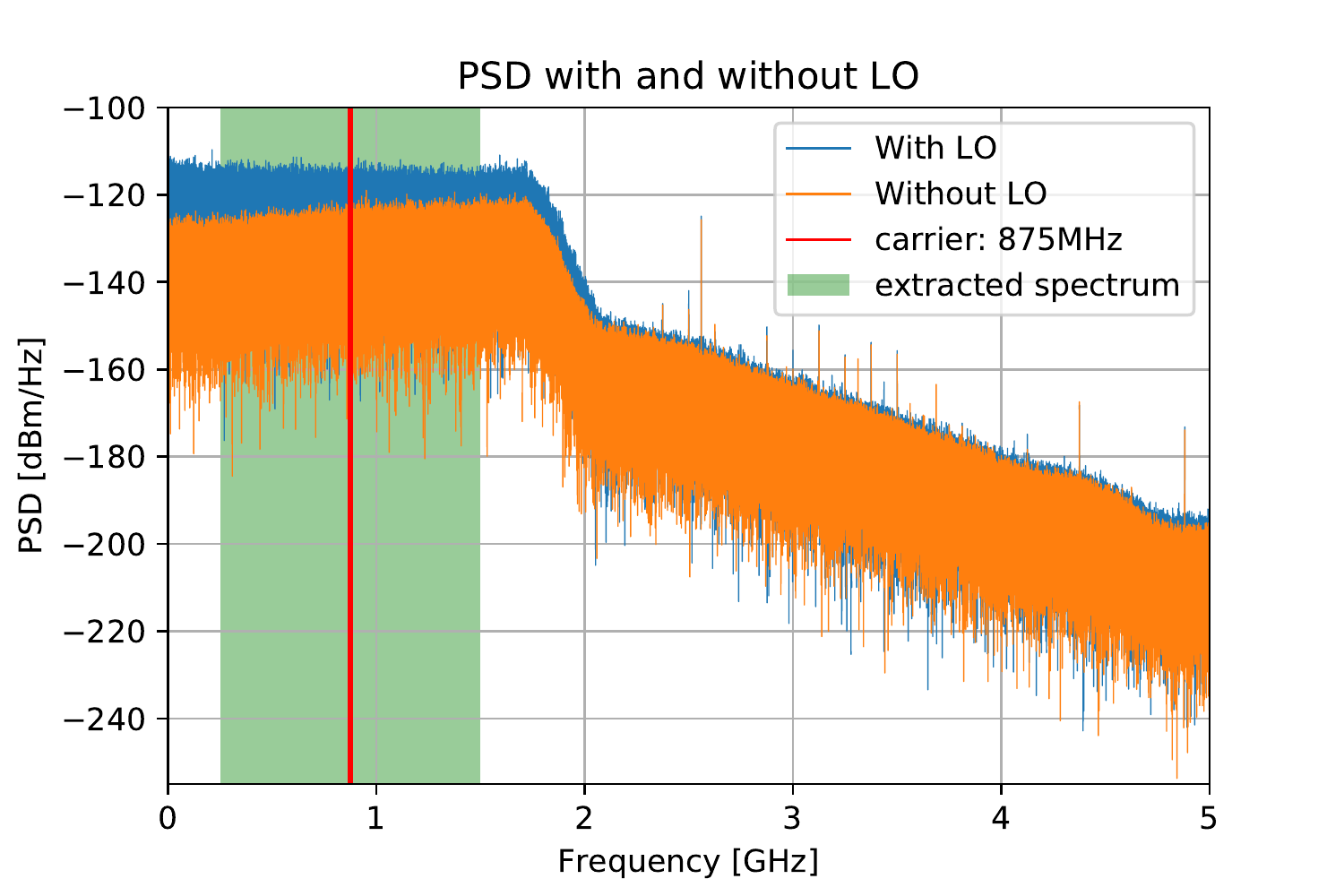}
	\end{center}
    \vspace*{-5mm}
	\caption{Spectrum obtained from the detectors with or without the LO  active. In green is highlighted the portion kept after the digital filtering and used for the generation. The peaks present after the $3\si{\decibel}$ point of the detectors are introduced by the oscilloscope at harmonics of the sampling frequency and are not present if the spectrum is obtained with an analog spectrum analyzer (HP 8561B).}
	\label{img:psd}
\end{figure}
To further reduce the classical noise from the detectors (at the expense of a reduced generation rate) we perform a filtering of the signal, as it can be see in the full schematic of the setup presented in Figure \ref{img:schema2D}.

Figure \ref{img:psd} shows the power spectral density of the signal produced by the detectors when the LO is turned on and when the LO is off. 
Although, the response seems uniform along the entire bandwidth of the detectors ($1.6 \si{\giga\hertz }$), the initial part of the spectrum ($DC - 1 \si{\mega\hertz}$) is affected by technical noise. In order to filter out this noise and enhance the signal-to-noise ratio, we have considered for the random generation only a window large $1.25 \si{\giga\hertz}$ centered around $875 \si{\mega\hertz}$. With this selection, the gap is never lower than $9.6 \si{\decibel}$.
The selection has been done digitally.
\begin{figure}[!htbp]
	\begin{center}
		\includegraphics[width=\linewidth]{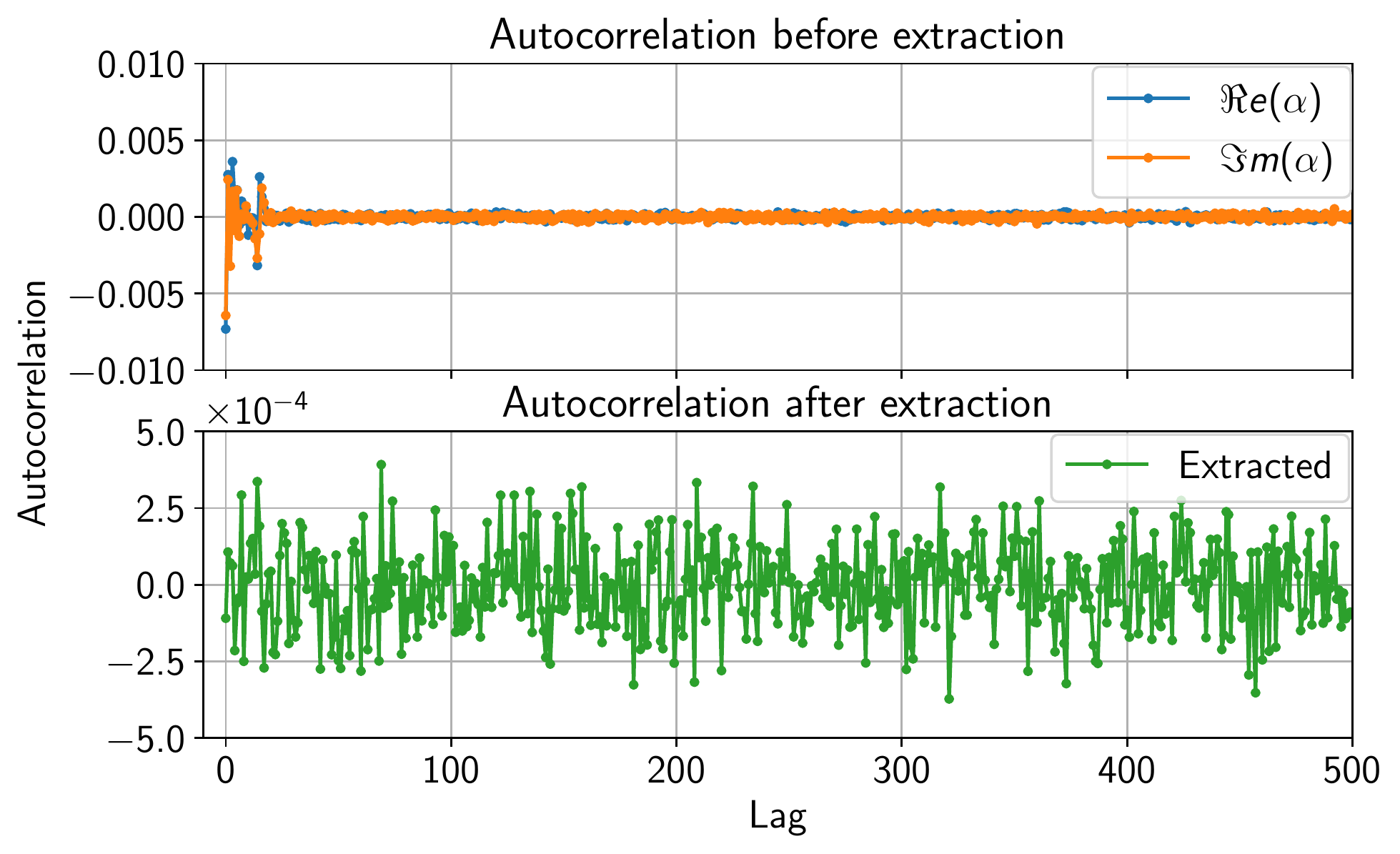}
	\end{center}
    \vspace*{-5mm}
	\caption{Autocorrelation measured for a sample of $5\cdot10^7$ filtered and extracted numbers. The spikes present in the first lags before the extraction are  due to the noise introduced by our sampling equipment. However, they are completely absent after the extraction. }
	\label{img:autocorrelation}
\end{figure}

However, employing a Brick-wall filter in the frequency domain, inevitably induces correlation in the time-domain of our signal: indeed we observe a $sinc$ dependence in the autocorrelation, as expected from the Wiener-Khinchin theorem. The correlation is removed by undersampling the signal in such a way to match the first zero of the autocorrelation function. 
Figure \ref{img:autocorrelation} shows the residual autocorrelation after the downsampling,  before and after the randomness extraction for a run of $5\cdot10^7$ samples. The results, even before the extraction, are good, with values always  below $7.5\cdot10^{-3}$ and typically below $1\cdot10^{-4}$, except for the first lags. The value of the first lag is due to noise introduced by the oscilloscope at harmonics of its sampling rate frequency.

In Figure \ref{img:psd}, these distortions are clearly visible at high frequencies, where there is no contribution from the signal.
However, after the extractor, all the classical noise is eliminated and the autocorrelation is completely flat, also for the initial lags.

\section{Statistical Tests}
\label{sec:statistica_tests}
In order to check for problems in our implementation we performed some statistical test on the generated numbers. First, we implemented the fast computable two-universal hash function introduced in \cite{Frauchiger2013}, then we used it to extract the final numbers from the filtered samples. We calibrated the extractor with the value of $H_{\rm min}(X|\mathcal{E})_\text{min}$ of eq. \eqref{eq:hmin} and then we extracted $\approx 5.18\cdot10^{10}$ random numbers from an initial set of $7.5\cdot10^{10}$ raw numbers. We tested them with the NIST \cite{bassham2010sp} and {\it dieharder} suite \cite{brown-2013}: in both cases all the tests were passed, as we can see in Table \ref{tab:test}. Passing these tests doesn't certify the randomness, but only shows that some patterns are not present in the analyzed data. However, since our QRNG is supposed to pass all of them, is a way to double-check that our setup is working as expected.

\onecolumngrid

\begin{table}[!htbp]
\begin{minipage}{0.4\linewidth}
\centering
\begin{tabular}{lll}
\hline
\rule{0pt}{3ex}Test's name & P-value & Result \\ 
\hline 
\rule{0pt}{3ex}diehard birthdays		&	0.398	&	  PASSED  	\\
diehard operm5			&	0.391	&	  PASSED  	\\
diehard rank 32x32		&	0.414	&	  PASSED  	\\
diehard rank 6x8		&	0.767	&	  PASSED  	\\
diehard bitstream		&	0.529	&	  PASSED  	\\
diehard opso			&	0.655	&	  PASSED  	\\
diehard oqso			&	0.758	&	  PASSED  	\\
	diehard dna			&	0.731	&	  PASSED  	\\
diehard count 1s str	&	0.482	&	  PASSED  	\\
diehard count 1s byt	&	0.361	&	  PASSED  	\\
diehard parking lot		&	0.515	&	  PASSED  	\\
diehard 2dsphere		&	0.484	&	  PASSED  	\\
diehard 3dsphere		&	0.739	&	  PASSED  	\\
diehard squeeze			&	0.580	&	  PASSED  	\\
diehard sums			&	0.140	&	  PASSED  	\\
diehard runs			&	0.478	&	  PASSED  	\\
diehard runs			&	0.316	&	  PASSED  	\\
diehard craps			&	0.348	&	  PASSED  	\\
diehard craps			&	0.937	&	  PASSED  	\\
marsaglia tsang gcd		&	0.504	&	  PASSED  	\\
marsaglia tsang gcd		&	0.444	&	  PASSED  	\\
sts monobit				&	0.204	&	  PASSED  	\\
sts runs				&	0.716	&	  PASSED  	\\
sts serial				&	0.151	&	  PASSED  	\\
rgb bitdist				&	0.056	&	  PASSED  	\\
rgb minimum distance	&	0.043	&	  PASSED  	\\
rgb permutations		&	0.068	&	  PASSED  	\\
rgb lagged sum			&	0.019	&	  PASSED  	\\
\botrule
\end{tabular} 
\end{minipage}
\hspace{1cm}
\begin{minipage}{0.4\linewidth}
\centering
\begin{tabular}{lll}
\hline\rule{0pt}{3ex}
		Test's name & P-value & Result \\ 
		\hline
        \rule{0pt}{3ex}Frequency		&	0.980	&	PASSED	\\
			BlockFrequency	&	0.323	&	PASSED	\\
			CumulativeSums	&	0.819	&	PASSED	\\
			CumulativeSums	&	0.265		&	PASSED	\\
			Runs			&	0.187	&	PASSED	\\
			LongestRun		&	0.864	&	PASSED	\\
			Rank			&	0.372	&	PASSED	\\
			DFT				&	0.341	&	PASSED	\\
			NonOverlappingTemplate	&	0.016	&	PASSED	\\
			OverlappingTemplate		&	0.748	&	PASSED	\\
			Universal				&	0.381	&	PASSED	\\
			ApproximateEntropy		&	0.509	&	PASSED	\\
			RandomExcursions		&	0.315	&	PASSED	\\
			RandomExcursionsVariant	&	0.047	&	PASSED	\\
			Serial					&	0.318	&	PASSED	\\
			LinearComplexity		&	0.373	&	PASSED	\\
  			\botrule
\end{tabular} 
\end{minipage}
\caption{Result of Dieharder (left) and NIST (right) test suite on the extracted random numbers. In the case of multiple tests in a category, the smallest have been reported. }
\label{tab:test}
\end{table}
\twocolumngrid

\bibliography{biblio.bib}
\end{document}